\documentclass[a4paper, cleveref, cleveref, thm-restate]{lipics-v2021}
\pdfoutput=1 %uncomment to ensure pdflatex processing (mandatatory e.g. to submit to arXiv)
\hideLIPIcs  %uncomment to remove references to LIPIcs series (logo, DOI, ...), e.g. when preparing a pre-final version to be uploaded to arXiv or another public repository

\usepackage{xspace}
\usepackage{wasysym}
\usepackage{tikz}
\usetikzlibrary{calc}
\usepackage{intcalc}
\usepackage{listings}

\def\bit{bit}
\def\kibi{Ki}
\def\mebi{Mi}
\def\gibi{Gi}
\def\byte{B}
\def\milli{m}
\def\second{s}
\def\percent{\%}
\newcommand{\qty}[2]{#1\,#2}

\graphicspath{ {./figures/} }

\def\eg{e.g.,\xspace}
\def\ie{i.e.\xspace}
\def\etal{et al.\xspace}
\def\whp{whp}
\def\nproc{\ensuremath{\mathfrak{P}}}

\def\algo#1{\textsl{#1}\xspace}

\def\FY{\algo{FY}}
\def\FYFull{\algo{Fisher-Yates Shuffle}}

\def\SandShufFull{\algo{ScatterShuffle}}

\def\IPS{\algo{IpSc}}
\def\IPSFull{\algo{In-place Scatter}}
\def\RoughScatter{\algo{RoughScatter}}
\def\ParRoughScatter{\algo{ParRoughScatter}}

\def\IPDS{\algo{IpScShuf}}
\def\IPDSFull{\algo{In-Place ScatterShuffle}}
\def\PIPDS{\algo{PIpScShuf}}
\def\PIPDSFull{\algo{Parallel In-Place ScatterShuffle}}

\def\Oh#1{\ensuremath{\mathcal O \! \left(#1\right)}}

\title{Engineering Shared-Memory Parallel Shuffling \\ to Generate Random Permutations In-Place}

\titlerunning{In-Place Shared-Memory Parallel Shuffling}

\author{Manuel Penschuck}{Goethe University Frankfurt}{rip_shuffle@algorithm.engineering}{https://orcid.org/0000-0003-2630-7548}{Supported by the Deutsche Forschungsgemeinschaft (DFG) under grant ME~2088/5-1 (FOR 2975 — Algorithms, Dynamics, and Information Flow in Networks)}

\authorrunning{M. Penschuck}

\Copyright{Manuel Penschuck}



\category{} %optional, e.g. invited paper

\relatedversion{} 

\supplement{
	The source code, relevant raw data, and anlysis scripts can be found \url{https://algorithm.engineering/rip_shuffle.tar} (will be uploaded to Zenodo).
	To include this work in other Rust project, we recommend  \url{https://crates.io/crates/rip_shuffle}.
}

\nolinenumbers

\begin{document}
\maketitle
\begin{abstract}
Shuffling is the process of rearranging a sequence of elements into a random order such that any permutation occurs with equal probability.
It is an important building block in a plethora of techniques used in virtually all scientific areas.
Consequently considerable work has been devoted to the design and implementation of shuffling algorithms.

We engineer, ---to the best of our knowledge--- for the first time, a practically fast, parallel shuffling algorithm with $\Oh{\sqrt{n}\log n}$ parallel depth  that requires only poly-logarithmic auxiliary memory.
Our reference implementations in Rust are freely available, easy to include in other projects, and can process large data sets approaching the size of the system's memory.
In an empirical evaluation, we compare our implementations with a number of existing solutions on various computer architectures.
Our algorithms consistently achieve the highest through-put on all machines.
Further, we demonstrate that the runtime of our parallel algorithm is comparable to the time that other algorithms may take to acquire the memory from the operating system to copy the input.
\end{abstract}

\section{Introduction}
Random permutations are heavily studied in many fields of science with numerous applications.
They are commonly considered an ``easy and fair'' arrangement and thus influence many aspects of everyday life ranging from shuffling a deck of cards in a friendly game to determining the fateful order in which soldiers are drafted for war (\eg \cite{exec_order_11497}).

In computer science, applications include numerical simulations,  sampling of complex objects, such as random graphs, machine learning, or statistical tests (\eg \cite{DBLP:journals/jct/BenderC78,DBLP:journals/corr/abs-2003-00736,NIPS2013_c3c59e5f,DBLP:journals/tomacs/Lemire19}).
Especially, if coupled with rejection sampling, shuffling can become a dominating subroutine (\eg \cite{DBLP:conf/alenex/Allendorf0PTW22} in which a prototype of this work was used without discussing any algorithmic details).
Further, the assumption that an input is provided in random order (instead of adversarially) allows for practical algorithms that are almost always efficient.
Among others, this notion motivates the random-order-model for online algorithms~\cite{ferguson1989solved}.
For the same reason, implementations of offline algorithms may start by shuffling their inputs;
for instances, folklore suggests to shuffle the input before sorting it with a simple \algo{Quicksort} implementation.

From an algorithmic point of view, the tasks of \emph{shuffling} and \emph{sorting} are tightly connected since both require an algorithm capable of emitting any permutation.
Though, while sorting needs to handle adversarial inputs, shuffling can be optimized for the well-behaved uniform distribution.
Shuffling can be implemented in linear-time via integer sorting by augmenting each input element with a uniform variate and sorting by it~\cite{DBLP:conf/ISCApdcs/CongB05}; we refer to this approach as \algo{SortShuffle}.
The famously impractical \algo{BogoSort} demonstrates the other direction, namely sorting by shuffling, but suffers from a ``slightly'' suboptimal expected runtime of $\Omega(n \cdot n!)$~\cite{DBLP:conf/fun/GruberHR07}.

The quest for \emph{in-place} algorithms is driven by the various costs of memory.
The most obvious aspect is that the maximal data set size that can be handled by a machine roughly halves if the output is produced in a copy.
Further, it takes a considerable time to allocate main memory on modern computer systems; in \cref{sec:evaluation} we demonstrate that the runtime of our shuffling algorithm is comparable to the time it takes to acquire additional memory of input size.
Another hidden cost is the increased code complexity to handle failed allocations of dynamic memory (\eg because the system ran out of memory). 
Finally, some programming languages have a concept of non-copyable data; \eg in C++ the copy-constructor can be deleted, and in Rust data types need to explicitly declare that they are copyable.
%Such types require algorithms which inherently swap objects.

\subsection{Our contributions}
We design and implement the practical shared-memory parallel algorithm \PIPDSFull (\PIPDS).
\PIPDS has a parallel depth of $\mathcal O (\log (n) \sqrt{n k / \log k})$ and uses $\Oh{n \log_k (n)}$ work (whp) where $k$ is small tuning parameter.
The algorithm is an in-place modification of \SandShufFull~\cite{DBLP:journals/ipl/Sanders98} and incorporates ideas of \algo{MergeShuffle}~\cite{DBLP:conf/gascom/BacherBHL18} to further improve the hidden constants.
We provide fast shuffle implementations in a free and well-tested plug-and-play Rust library.
Our \PIPDS does not use heap allocations and emits reproducible permutations if a seedable pseudo-random number generator is provided.

After a discussion of notation and related work in \cref{sec:preliminaries,sec:sequential}, we derive the sequential \IPDSFull (\IPDS) in \cref{sec:sequential} and parallelize it in \cref{sec:parallel}.
In \cref{sec:implementation}, we discuss details of our implementations which we then evaluate in \cref{sec:evaluation}.

\section{Preliminaries and notation}\label{sec:preliminaries}
The expression $(x_i)_{i=a}^b$ denotes the sequence $x_a, \ldots, x_b$ and may be shortened to $(x_i)_i$ if the limits are implied by context.
We indicate an array of $n$ elements as $X[1..n]$ and reference the subrange $X[i], \ldots, X[j]$ as $X[i..j]$.
Further, $[n]$ denotes the set $\{1, \ldots, n\}$.
Then, a permutation is a bijection $\pi\colon [n] \to [n]$ where $\pi(i)$ encodes the position of the $i$-th input element in the output.
We say that a probabilistic statement holds \emph{with high probability} (whp) if the error probability is at most $1 / n$ for some implied parameter $n$.

\subsection{Parallel model of computation}
For parallel algorithms, we assume the commonly accepted binary Fork-Join model~\cite{DBLP:books/daglib/0023376}.
This choice fits the \texttt{rayon}\footnote{\url{https://crates.io/crates/rayon}} infrastructure used in our implementation well.
An execution starts with a single task on a unit-cost random access machine.
Additionally, any task $t_0$ can recursively \emph{fork} into two tasks~$t_1$ and~$t_2$.
In this case~$t_0$ waits until $t_1$ and $t_2$ complete their computation and \emph{join} to resume~$t_0$.
In practice, Fork-Join frameworks, such as oneTBB\footnote{previously known as Intel Threaded Building Blocks, \url{https://github.com/oneapi-src/oneTBB}}, Cilk\footnote{see also \url{https://www.opencilk.org}, \url{http://cilkplus.org}} \cite{DBLP:conf/irregular/Leiserson97}, or rayon, use a worker-pool in combination with a work-stealing scheduler to map tasks to cores.
Algorithmic performance measures are the \emph{work}, \ie the total number of instructions, and the \emph{parallel depth} (\emph{span}), defined as the length of the critical path which corresponds to the execution time assuming an unbounded number of workers.

\subsection{Random shuffling}
The sequential \algo{Fisher-Yates-Shuffle} (\FY, also know as \algo{Knuth-Shuffle})~\cite{DBLP:books/aw/Knuth81} obtains a random permutation of an array $A[1 .. n]$ in time $\Oh{n}$.
Conceptually, it places all items into an urn, draws them sequentially without replacement, and returns the items in the order  they were drawn.
The algorithm works in-place and fixes the value of $A[i]$ in iteration~$i\in [1...n{-}1]$ by swapping $A[i]$ with $A[j]$ where~$j$ is chosen uniformly at random from the not yet fixed positions $[i..n]$.
In other words, in the $i$-th iteration, the $(i{-}1)$-prefix of $A$ stores the result obtained so far, while the $(n{-}i)$-suffix represents the urn.

Shun \etal show that this seemingly inherently sequential algorithm exposes sufficient independence to be processed with logarithmic parallel depth (\whp)~\cite{DBLP:conf/soda/ShunGBFG15}.
Later, Gu \etal propose an in-place variant based on the so-called decomposition property of the parallel \FY~\cite{DBLP:conf/apocs/GuOS21}.
However, both algorithms are designed to solve a subtly different problem.
They \emph{permute} the input in an explicitly prescribed manner.
As such, the permutation is part of the input and the implementation\footnote{\url{https://github.com/ucrparlay/PIP-algorithms} master at time of writing (\texttt{6af1df9})} of \cite{DBLP:conf/apocs/GuOS21} uses two additional pointers per element (\ie shuffling \qty{32}{\bit} values on a \qty{64}{\bit} machine leads to a five-fold increase of memory).

A random permutation can be computed in parallel by~\nproc{} processors by assigning each element to one of~\nproc{} buckets uniformly at random and then applying the sequential algorithm to each bucket~\cite{DBLP:journals/ipl/Sanders98}.
We refer to this algorithm as \SandShufFull and build on it in \cref{sec:sequential}.
A similar technique yields an I/O-efficient random permutation algorithm~\cite{DBLP:journals/ipl/Sanders98}.

Going the opposite direction also yields an efficient algorithm.
\algo{MergeShuffle} (\algo{MS}) first assigns each processor a contiguous section of the input array, shuffles the subproblems pleasingly parallel and finally recursively merges them~\cite{DBLP:conf/gascom/BacherBHL18}.
Merging can be interpreted as the inverse of \SandShufFull's scatter with two buckets.
In a precursor study, we found it too slow to generalize \algo{MergeShuffle} using $k$-way merging which is needed to reduce the recursion depth.
However, our \RoughScatter routine in \cref{subsec:rough-shuffle} builds on the insight that binary random merging can be implemented using a single random bit for most elements.

Cong and Bader~\cite{DBLP:conf/ISCApdcs/CongB05} empirically study additional techniques such as shuffling using sorting algorithms (\algo{SortShuffle}) or random dart-throwing (\algo{DartThrowingShuffle}).
We are, however, unaware of how to implement these approaches in-place.\footnote{
	Our \IPDS algorithm can be interpreted as an optimized in-place RadixSort in which buckets are randomly drawn.
	As such, there are conceptual similarities to \algo{SortShuffle}.
}

\subsection{Sampling from discrete distributions}
In the following, we sample from several discrete probability distributions (arguably, shuffling is just that).
This is achieved by first obtaining a stream of independent and unbiased random bits that are subsequently reshaped to attain the required distribution.
The default way of implementing the first step is using a pseudo-random generator, such as Pcg64Mcg~\cite{oneill:pcg2014}.
%To compare to MergeShuffle~\cite{DBLP:conf/gascom/BacherBHL18}, we also adopt the \textsc{rdrand} instruction~\cite{7477490} included in modern \textsc{x86} processors; depending on the specific processor it one to two orders of magnitude slower than MT19937-64.~\cite{agner-fog-instr-tab,Route_2017}

Sampling an integer from $[0, s)$ with $s=2^k$ for some $k \in \mathbb N$ from random words is very cheap and involves only shifting and masking.
We adopt rejection-based algorithms with expected constant time to sample uniform variates from $[0, s]$ for general $s$ (see~\cite{DBLP:journals/tomacs/Lemire19}) and binomial variates (see~\cite{DBLP:books/sp/Devroye86}).
Sampling of $k$-dimensional multinomial variates is implemented by chaining appropriately parametrized binomial samples in expected time $\Oh{k}$.

\section{Sequential in-place shuffling}\label{sec:sequential}
\begin{figure}
\begin{lstlisting}[tabsize=2,basicstyle=\footnotesize\tt,numbers=left, frame=single, backgroundcolor=\color{white}]
pub fn fisher_yates<R: Rng, T>(rng: &mut R, data: &mut [T]) {
  for i in (1..data.len()).rev() {
    data.swap(i, uniform::gen_index(rng, i+1)); // partner from [0..i]
}}\end{lstlisting}

	\caption{Fisher Yates implementation in Rust}
	\label{alg:fisher-yates}
\end{figure}

In this section, we propose \IPDSFull (\IPDS), a sequential in-place variant of Sanders' parallel \SandShufFull~\cite{DBLP:journals/ipl/Sanders98}.
\Cref{alg:ipds} summarizes the algorithm.
Building on the performance results obtained, we reintroduce parallelism
in \cref{sec:parallel}.

\subsection{State of the art}
The shuffle algorithm most commonly used  in practice seems to be the simple and fast \FYFull (\FY)~\cite{DBLP:books/aw/Knuth81}.
In fact, \cref{alg:fisher-yates} is the exact ``na\"ive'' implementation used in \cref{sec:evaluation} and nearly identical to the version included in the \texttt{rand} crate, the de facto standard randomization library in the Rust ecosystem.
Due to its simplicity, the algorithm outperforms the more advanced schemes for small inputs.
However, \FY's unstructured accesses to main memory cause a severe slowdown for larger inputs. %, that warrants more cache-friendly algorithms.
This is especially relevant for parallel algorithms where the memory subsystem is shared between cores (see \cref{sec:evaluation}).

\SandShufFull  is designed to be a parallel algorithm that also fares well in the external memory model~\cite{DBLP:journals/ipl/Sanders98}.
Given an input $(x_i)_{i=1}^n$, the algorithm moves each input element~$x_i$ into a bucket drawn independently and uniformly from $B_1, \ldots, B_k$.
Afterwards, each bucket constitutes an independent subproblem of expected $\Theta(n/k)$ elements on which we recurse.
%We can recurse on each bucket and finally output the elements.
For small subproblems, we switch to \FY as the base case algorithm.

In the original parallel formulation of \SandShufFull, the number of buckets~$k$ equals the number of processing units~$\nproc$ to expose the maximal degree of parallelism.
Sanders, however, already discusses that in the presence of memory hierarchies, the parameter~$k$ should be chosen sufficiently small such that the individual processors can cache at least the tail of each bucket.
In his implementation\footnote{\url{https://web.archive.org/web/20050827081959/http://www.mpi-sb.mpg.de/~sanders/programs/randperm/} [sic]}, the parameter is $k = 32$ for the largest runs in~\cite{DBLP:journals/ipl/Sanders98}.
At time of writing ---more than two decades later, using very different hardware to run experiments with more than three orders of magnitude larger data sets--- we empirically find $k \le 256$ to be the best choice for our $\IPDSFull$ over a wide range of input sizes.
Hence, $k$ should be intuitively treated as a small constant that governs primarily the branching factor of the recursion.
In \cref{sec:parallel}, we will add parallelism independent of $k$.

Our main modification leading to \IPDS is \IPSFull (\IPS) which scatters the input into $k$ buckets.
Formally, let $X = (x_i)_{i=1}^n$ be the input and $A=(a_i)_{i=1}^n$ be independent uniform variables from $[1, k]$ indexing into the aforementioned buckets.
Then, \IPS groups $X$ by $A$ by rearranging the elements in $X$ with some permutation~$\pi$ that sorts $A$.
%The freedom to arbitrarily rearrange elements within a bucket, will be exploited later for optimizations.

Similar problems have been studied in the context of integer sorting.
The special case of $k = 2$ (\ie binary partition) and $k = 3$ (known as the \emph{Dutch national flag problem}) can be efficiently solved in-place~\cite{DBLP:books/ph/Dijkstra76, meyerfailure}.
For $k>3$, two-pass approaches can be used (\eg \emph{American flag sort} \cite{DBLP:journals/csys/McIlroyBM93}) but require repeated access each $a_i$.
The parallel implementation~\cite{DBLP:conf/icse/SinglerK08} of \SandShufFull included in \texttt{libstdc++} uses this technique and stores $A$ explicitly requiring $\Theta(n \log k)$ bits.
Another way, in the spirit of \cite{DBLP:journals/jpdc/FunkeLMPSSSL19}, is to require a pseudo-random generator that can be replayed multiple times by copying and retrieving the generator's internal state.

\subsection{IpScShuf --- an in-place implementation of ScatterShuffle}\label{subsec:ipscs}
\begin{figure}
	\begin{lstlisting}[tabsize=2,basicstyle=\footnotesize\tt,numbers=left, frame=single, backgroundcolor=\color{white}]
		pub fn shuffle(&self, rng: &mut R, data: &mut [T]) {
			if data.len() <= self.config.seq_base_case_size() {
				return self.config.seq_base_case_shuffle(rng, data); 
			}
			let recurse = |rng: &mut R, dat: &mut [T]| self.shuffle(rng, dat);
			let mut buckets = split_slice_into_equally_sized_buckets(data);
			
			rough_scatter(rng, &mut buckets); // assign most items to random buckets
			
			{ // FineScatter: assign the remaining items
				let num_staged = buckets.iter().map(|b| b.num_staged()).sum();
				let final_sizes = sample_final_sizes(rng, num_staged, &buckets);
				move_buckets_to_fit_final_sizes(&mut buckets, &final_sizes);
				shuffle_stashes(rng, &mut buckets, recurse);
			}
			
			// recurse on the buckets
			buckets.iter_mut().foreach(|b| recurse(rng, b.data_mut()));
		}	
	\end{lstlisting}

	\caption{Implementation of \IPDSFull in Rust up to renaming.}
	\label{alg:ipds}
\end{figure}
In the following, we describe \IPSFull (\IPS) that supports true random bits, and, whp, runs in linear time using only $\Oh{k \log k}$ bits additional storage.
Since each of the $n$ items is assigned a uniformly selected bucket, the numbers $(n_i)_{i=1}^k$ of elements assigned to each bucket follow a multinomial distribution and are tightly concentrated around $n / k$.

For the remainder, we assume that $n \gg k (\log k)^3$ and $k^3 \log k = \Oh{n}$, since otherwise, the problem is so small that \FYFull is more appropriate.
These assumptions are only needed to bound \IPDS's complexity and do not affect its correctness.
In practice, they translate to a minimal recommended size of roughly $10^5$ elements.

A straightforward solution is to draw the sizes of all buckets as a multinomial random variate.
We then sample the buckets without replacement weighted by their decreasing target size.
This can be implemented in expected linear time using suitable dynamic weighted sampling data structures (\eg \cite{DBLP:journals/mst/MatiasVN03}).
As discussed further in \cref{sec:implementation}, such approaches are outperformed by the following scheme.
Inspired by the framework of \cite{DBLP:conf/gascom/BacherBHL18} (but quite different in its details), our assignment task consists of two phases.
Firstly, during the \RoughScatter phase, we very efficiently assign the vast majority of items --- but almost certainly not all of them.
Secondly, during the \algo{FineScatter} phase, we process the remaining few elements.

\subsection{RoughScatter --- the opportunistic work horse}\label{subsec:rough-shuffle}
\RoughScatter exploits the aforementioned concentration of the final bucket sizes around their mean of $n / k$ to assign elements in an opportunistic fashion until we hit said $n/k$ barrier. % as follows.
Let $X[1..n]$ denote the input array.
As illustrated in \cref{fig:buckets}, we partition $X$ into $k$ contiguous buckets of equal sizes $n / k$ (up to rounding).
Each bucket $B_i$ is stored as a triple of indices $(b_i, s_i, e_i)$ where $b_i$  points to the beginning of $B_i$ and $e_i$ beyond the bucket's end.
A bucket is further subdivided into (i) an initially empty segment $X[b_i..s_i)$ of so-called \emph{placed} items and (ii) $X[s_i..e_i)$ of so-called \emph{staged} items.
We say that bucket $B_i$ is \emph{full} iff all items are placed, i.e. $s_i = e_i$.
Up to the last step in \cref{subsec:fine-scatter}, \emph{staged} items can be freely moved around, whereas the position of \emph{placed} items carries meaning.

\begin{figure}
	\begin{center}
		\scalebox{0.8}{
			\begin{tikzpicture}
				\def\bWidth{30mm}
				\def\bHeight{6mm}
				\def\bFill{0.6}
				\def\numBuckets{3}
				
				\foreach \j in {1 ,..., \numBuckets} {
					\node[fill=red!30, minimum height = \bHeight, anchor=south west, minimum width=\bWidth] at (\j * \bWidth - \bWidth, 0)  {};	
					\node[fill=green!10, minimum height = \bHeight, anchor=south west, minimum width=\bFill * \bWidth] at (\j * \bWidth - \bWidth, 0)  {};	
					\path[draw] (\bWidth * \j - \bWidth + \bFill * \bWidth, 0) to ++(0, \bHeight);
				}

				\foreach \j in {0 ,..., \numBuckets} {
					\path[draw, thick] (\bWidth * \j, 0) to ++(0, \bHeight);
				}
				
				\path[draw, thick] (0, 0) to ++(\numBuckets * \bWidth, 0);
				\path[draw, thick] (0, \bHeight) to ++(\numBuckets * \bWidth, 0);
				
				\node[anchor=south west, minimum width=\bWidth] at (\bWidth, \bHeight) {Bucket $B_i$};	
				
				\node[anchor=south west, minimum width=\bWidth+1em, minimum height=\bHeight+2em, fill=white, opacity=0.8] at (-1em-0.5pt,-1em) {};

				\node[anchor=south west, minimum width=\bWidth+1em, minimum height=\bHeight+2em, draw=white, fill=white, opacity=0.8] at (2*\bWidth +0.5pt,-1em) {};

				\foreach \x/\l in {0/$b_i$, \bFill/$s_i$, 1/$e_i$} {
					\path[draw, <-, thick] (\bWidth + \x * \bWidth + 0.05 *\bWidth, -0.1*\bHeight) to ++(0, -0.9em);
					\node[anchor=base, baseline, inner sep=0] at (\bWidth + \x * \bWidth + 1.1em, -1em) {\l};
				}
				
				\foreach \x in {0,..., \numBuckets0} {
					\path[draw, opacity=0.2] (\x * \bWidth / 10, 0) to ++(0, \bHeight);
				}
				
				\node[anchor=base, baseline, yshift=-0.25em] at (-2em, \bHeight / 2) {$\cdots$};
				
				\node[anchor=base, baseline, yshift=-0.25em] at (\bWidth + \bWidth * \bFill / 2, \bHeight / 2) {\small placed};
				\node[anchor=base, baseline, yshift=-0.25em] at (1.5 * \bWidth + \bWidth * \bFill / 2, \bHeight / 2) {\small staged};		
				
				\node[anchor=base, baseline, yshift=-0.25em] at (\numBuckets * \bWidth + 2em, \bHeight / 2) {$\cdots$};

		\end{tikzpicture}}
	\end{center}

	\caption{
		\IPDS partitions the input into $k$ buckets, each roughly containing $n / k$ elements.
		Initially, all items are \emph{staged} ($b_i = s_i$) and the bucket is said to be \emph{empty}.
		Eventually, more and more items are placed (from the left). If $s_i = e_i$ the bucket is said to be \emph{full}.
	}
	
	\label{fig:buckets}
\end{figure}
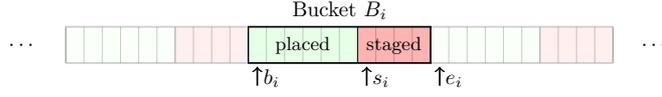

In each iteration, the algorithm finalizes the bucket assignment of the element~$x$ that $s_1$ points to, \ie the first staged element in $B_1$ at that point in time.
To this end, we randomly draw a partner bucket $j$ uniformly from $[1..k]$, swap the elements $X[s_1] \leftrightarrow X[s_j]$ (skipped if $j=1$), and increment $s_j$.
As a result, element~$x$ is moved into the \emph{placed} region of the partner bucket~$B_j$.
If $B_j$ is now full, the algorithm stops, otherwise it repeats.

\begin{lemma}
	Let $P$ be the set of elements placed by \RoughScatter and let $x \in P$ be an arbitrary placed item.
	Then $x$ is assigned to bucket $B_i$ with probability~$1 / k$.
\end{lemma}

\begin{proof}
	With loss of generality, we assume that initially all items are staged. %; this is in line with the description above, but not strictly necessary.
	Then, there is a unique iteration for each item $x \in P$ in which it gets placed.
	To this end, the then still staged element~$x$ is swapped with a staged item~$y$ where $y \in B_i$ with probability~$1/k$.
	It then increases $s_i$ and thereby defines $x$ as placed.
	Since \RoughScatter only swaps staged items, this placement of $x$ is final.
	The possible change of position of element $y$ is inconsequential, since each assignment is carried out independently.
\end{proof}

\noindent
In the following, we bound the number elements that remain staged after \RoughScatter.

\begin{lemma}\label{lem:r_bound}
	After a \RoughScatter run, let $r_i = e_i - s_i$ be the number of items still staged in bucket $b_i$ and $R = \sum_i r_i$ their sum.
	For $n \gg k (\log k)^3$, we have $R \le \sqrt {2 n k \log k}$ whp.
\end{lemma}

\begin{proof}
	We interpret the input to \RoughScatter as $n$ \emph{balls} that are independently thrown into $k$ uniform bins.
	If we run the balls-into-bins experiment to completion, the maximal load of any bucket is at most $M(n, k) = \frac n k + \sqrt{2 \frac n k \log k}$ whp~\cite{DBLP:conf/random/RaabS98}.\footnote{A similar argument was already used in the analysis of \SandShufFull~\cite{DBLP:journals/ipl/Sanders98}.}

	Let $n'$ be the number of balls assigned in said game when the maximal load first reached $n/k$.
	Algorithmically, this corresponds to the termination of \RoughScatter.
	By identifying $M(n', k) = n / k$ and solving for $n'$, we find that whp $n' \ge n - \sqrt{2nk \log k} := n -R$. % which shows $R \le \sqrt {2 n k \log k}$ with high probability.
\end{proof}

\begin{remark}
	The fraction of unprocessed elements $R / n$ vanishes for $n \to \infty$. 
	Even for small inputs with $n = 2^{22}$ and practical $k=64$, less than $1 \%$ of the input remains unassigned whp.
\end{remark}

\iffalse
	\def\balmult#1#2{\ensuremath{\mathcal{BM}(#1, #2)}}
	\let\follows\sim

	To bound the $r_i$, we first recall properties of the multinomial distribution. 
	Let $N = (n_1, \ldots, n_k)$ describe the number of balls in $b_i$ after throwing $n$ balls in $k$ equally likely bins.
	Then $N$ follows $\balmult n k$ where $\balmult n k$ is the \emph{balanced multinomial distribution} in which each event has probability $p_i = 1 / k$.
	Let $N \follows \balmult n k$ and $M \follows \balmult m k$ be two random variables governed by two balanced multinomial distributions.
	Then it is easy to check that $N + M \follows \balmult {n+m} k$.
	
	Thus, we can interpret the experiment of throwing $n$ balls into $k$ bins to completion in two steps.
	We first consider the $n- R$ throws until the first bin contains $n / k$ balls.
	Let $F \follows \balmult {n-R} k$ denote sizes of bins after this first phase.
	Then, we continue with the remaining $R$ balls and denote the final bin sizes as $N \follows \balmult n k$.
	Let $S = N - F$ the number of balls added to each bin in the second phase. 
	From the closure under addition of $\balmult \cdot k$, we have that $S \follows \balmult R k$.
\fi

\subsection{FineScatter --- fixing the small remainder}\label{subsec:fine-scatter}
After the execution of \RoughScatter only $R = \Oh{\sqrt{n k \log k}}$ items need to be assigned during the \algo{FineScatter} phase whp.
If our initial assumption still holds for $R \gg k (\log k)^3$, we can compact the staged items into a contiguous memory area, apply \RoughScatter and recurse. 
However, for small inputs the assumption is likely violated, while for large inputs the fraction $R / n$ contributes only negligibly to the total runtime.
Thus, we do not consider it worthwhile to devise a merging procedure for this case and instead directly use a dedicated base case algorithm based on the following Lemma.
It lays out the route to efficiently obtain the final bucket sizes and independently assign the remaining elements.

\begin{lemma}	\label{lem:multinomial_shuffle}
	Let $X = (x_i)_{i=1}^n$ be a sequence and  $N = (n_i)_{i=1}^k$ be sampled from a multinomial distribution with equal weights $p = 1/k$ such that $\sum_i n_i = n$.
	Let $f_N\colon [n] \to [k]$ be an arbitrary partition of $X$ with class sizes $N$.
	Finally, let $\pi\colon [n] \to [n]$ a random permutation.
	Then, for fixed $i$ and $j$, the probability that element $x_i$ is mapped by $f_N(\pi(i))$ to class $j$ is $1/k$.
\end{lemma}

\begin{proof}
	Due to symmetry, it suffices to consider the first partition class $j = 1$.
	Its size $n_1$ follows a binomial distribution over $n$ attempts with success probability $p = 1/k$ by definition of the multinomial distribution.
	Additionally, there exists a permutation $\gamma$ such that the composition $f_N \circ \gamma$ maps the indices $1, \ldots, n_1$ to the first partition class.
	Since $\pi$ is chosen uniformly at random, $\pi' = \pi \circ \gamma$ is equally likely. 
	Thus, it suffices to compute the total probability that $\pi'$ puts a fixed $x_i$ into the first $n_1$ ranks over all $0 \le n_1 \le n$:
	\begin{align}
		\sum_{j=0}^{n} &P[\pi'(i) \le n_1 \, |\, n_1 = j ] \cdot P[n_1 = j] 
		=  \sum_{j=0}^{n} \underbrace{\frac j n}_{= 1 - (1 - \frac j n)} \, \cdot \, \binom{n}{j} \left( \frac 1 k \right) ^ j \left( 1 - \frac 1 k \right) ^ {n -j} \\
		&= \underbrace{\sum_{j=0}^{n} 1 \, \cdot \, \binom{n}{j} \left( \frac 1 k \right) ^ j \left( 1 - \frac 1 k \right) ^ {n -j}}_{=1}
		- \sum_{j=0}^{n} \underbrace{(1 - \frac j n) \, \cdot \, \binom{n}{j}}_{
			= \begin{cases}
				\binom{n-1}{j} & \text{if } j < n \\
				0              & \text{if } j = n
		\end{cases}}	
		\left( \frac 1 k \right) ^ j \left( 1 - \frac 1 k \right) ^ {n -j}\\
		&= 1 - \underbrace{\sum_{j=0}^{n - 1} \binom{n - 1}{j} \left( \frac 1 k \right) ^ j \left( 1 - \frac 1 k \right) ^ {(n - 1) - j}}_{=1} \left( 1 - \frac 1 k \right) ^ {1} 
		= 1/k		\qedhere
	\end{align}
\end{proof}

\subsubsection{Finalizing the bucket sizes}
Let $N=(n_i)_i$ be the numbers of elements assigned to bucket $B_i$ by \RoughScatter.
Guided by \cref{lem:multinomial_shuffle}, the base case algorithm first draws a multinomial variant $N' = (n'_i)_i$ where $n'_i$ corresponds to the number of elements that will be placed into bucket~$B_i$ by \algo{FineScatter}.
Then, the final sizes $N^f = (n^f_i)_i$ are $n^f_i = n_i + n'_i$.
Since $N$ and $N'$ follow a multinomial distribution with $k$ equally weighted classes, their sum $N^f$ does too.

By construction, the expected bucket size is $n / k$.
Let $d_i = n^f_i - n/k$ denote the deviation of the size of bucket $B_i$, i.e. the number of elements it needs to gain over the initial estimation of \RoughScatter. 
Analogously to the proof of \cref{lem:r_bound}, we bound  $\max_i\left\{|d_i|\right\} = \mathcal O(\sqrt{n/k \log k})$ whp~\cite{DBLP:conf/random/RaabS98,DBLP:journals/ipl/Sanders98}.
Thus, in all likelihood, the bucket boundaries only move slightly.

Luckily, \SandShufFull is oblivious to the order of elements within a bucket prior to recursion.
Thus, it suffices to appropriately move a few items near the boundaries of the buckets using our \algo{TwoSweep} algorithm.
First, we iterate over the buckets in ascending index order.
Meanwhile, we keep a counter $C_i = \sum_{j=1}^{i-1} d_j$ that indicates how many additional items are needed left of the current bucket~$B_i$.
If bucket $B_i$ is too large by more than $C_i$ items, we swap the excess staged items into the staging area of bucket~$B_{i+1}$.
In a second sweep from $B_k$ to $B_1$, we move the remaining excess items towards smaller bucket indices.

\begin{restatable}{lemma}{lemmnf}
	\label{lem:mnf}
	Let $N^f = (n^f_i)_i$ be the final bucket sizes, denote their deviation from the mean $n/k$ as $d_i = n^f_i - n/k$, and let $D_i = \sum_{j=1}^{i} d_j$ be the inclusive prefix sum of deviations.
	Then, \algo{TwoSweep} executes a total of $M(N^f) = \sum_i |D_i|$ swaps and takes time $\Oh{k + M(N^f)}$.
\end{restatable}

\begin{proof}
	\algo{TwoSweep} can exchange a staged item of bucket $B_i$ with its direct neighbors $B_{i \pm 1}$ in $\Oh{1}$ time by executing a single swap and adopting the pointers of the two involved buckets;
	exchanging an item between buckets $B_i$ and $B_j$ this way implies a chain of $|i - j|$ swaps causing $\Oh{|i-j|}$ work.
	Based on this, \algo{TwoSweep} carries out two snow-plow-like motions likely pushing intermediate items along the chain. For the remainder see \cref{app:proofs}.
	
	Observe that a positive value~$d_i$ indicates that bucket~$B_i$ needs to receive $d_i$ additional staged items from other buckets.
	Contrary, a negative value~$d_i$ means that bucket~$B_i$ has to give away $-d_i$ elements.
	The prefix sum $D_i$ has an analogous meaning but accumulated over the first $i$ buckets.
	This leads to the following cases:
	
	\begin{enumerate}
		\item 
		If $D_i$ is positive, buckets~$B_1, \ldots, B_i$ have an excess of $D_i$ items required somewhere in $B_{i+1}, \ldots, B_k$.
		These $D_i$ items will be pushed to the right during the first sweep.
		
		\item 
		If $D_i$ is negative, buckets~$B_1, \ldots, B_i$ have a demand of $|D_i|$ items met by an excess somewhere in the buckets~$B_{i+1}, \ldots, B_k$.
		Thus, $B_i$ receives $|D_i|$ items in the second sweep.
	\end{enumerate}
	
	\noindent
	In sum, bucket~$B_i$ is involved in $|D_i|$ swaps with its direct neighbors, leading to a total of $M(N^f)$ swaps and $\Oh{M(N^f) + k}$ work where the $k$ accounts for per-bucket overheads.
\end{proof}

\begin{restatable}{corollary}{twosweeptime}
	\label{cor:two-sweep-time}
	\algo{TwoSweep} takes time $\Oh{k\sqrt{nk \log k}}$ whp.
\end{restatable}

\begin{proof}
	We prove the claim based on \cref{lem:mnf} by establishing $\sum_i |D_i| = \Oh{k \sqrt{n k \log k}}$ (whp) where $D_i = \sum_{j=1}^i d_i$ is the prefix sum over the bucket size deviations from the mean.
	Observe that by construction, only elements that remain staged after the execution of \RoughScatter can contribute and therefore $\sum_i |d_i| \le 2 R$ where $R \le \sqrt{2 nk \log k}$ (whp) by \cref{lem:r_bound}.
	Additionally, since the deviations balance over all buckets we have $\sum_i d_i = 0$.
	Thus, we trivially have that $\max_i |D_i| \le R$.
	
	We assume a worst case deviation where the first bucket needs to gain all $R$ elements from the last bucket (or vice versa).
	While this is rather pessimistic (recall $\max_i\left\{|d_i|\right\} = \Oh{\sqrt{n/k \log k}}$  whp), it suffices to show the bound.
	In this instance, we have $|D_i| \le R$ for all $i$ and $\sum_{i=1}^k |D_i| \le \sum_{i=1}^k R \le k R = \Oh{k \sqrt{n k \log k}}$ (whp). 
\end{proof}

\subsubsection{Assigning the remaining staged elements} 
At this point in the execution, all buckets have reached their final sizes, but each bucket~$B_i$ still has $n'_i$ staged items with $\sum_i n'_i = R$.
Rather than sampling weighted by $N' = (n'_i)_i$, we apply \cref{lem:multinomial_shuffle} and instead randomly shuffle all staged items.
This can be done by compacting the staged item into $X[1..R]$ and shuffling $X$.
To this end, we swap the staged items with the items originally stored in $X$; after shuffling, we apply the same swap sequence in reverse to restore the original items and put the staged items into a now random permutation.

\subsection{Putting it all together}
\IPSFull (\IPS) is the algorithm executed on each recursion layer of \IPDS. It runs \RoughScatter and \algo{FineScatter} in sequence to randomly assign $n$ elements to $k$ buckets.
The input is rearranged such that each bucket corresponds to a contiguous memory region.

\goodbreak

\begin{lemma}\label{lem:runtime-without-recursion}
	For $n \gg k \log^3k$ and $k^3 \log k = \Oh{n}$, \IPS assigns $n$  items in time $\Oh{n}$ whp.
\end{lemma}

\begin{proof}
	We sum up the four tasks carried out:
	\begin{enumerate}
		\item
		 	\RoughScatter first partitions the input into buckets in time $\Oh{k}$ and then randomly assigns $n - R = \Oh{n}$ elements whp.
			Assuming that obtaining a word of randomness takes constant time, this translates into a time complexity of $\Oh{k + n} = \Oh{n}$.
			
		\item 
			Sampling a $k$-dimensional multinomial random variate takes time $\Oh{k}$ whp.
			
		\item 
			Running \algo{TwoSweep} to adjust the boundary size takes time $\Oh{k \sqrt{n k \log k}} = \Oh{n}$ whp.
			
		\item 
			Shuffling the staged items with \FYFull takes time $\Oh{R} = \Oh{n}$ time.\footnote{Based on \cref{thm:ipds-runtime}, we are also free to recurse with \IPDS instead of using \FYFull.}\qedhere
	\end{enumerate}
\end{proof}

\goodbreak

\noindent 
\IPDSFull consists of recursive applications of \IPS.
We stop the recursion on a subproblem as soon as it reaches the base case size of $N_0 = \Oh{1}$ at which point it is finalized using \FYFull.

\begin{theorem}\label{thm:ipds-runtime}
	With high probability \IPDSFull takes time $\Oh{n \log_k (n / N_0)}$ and $\Oh{k \log_k (n / N_0)}$ additional words of storage where $N_0 = \Omega(k^3)$ is the base case size.
\end{theorem}

\begin{proof}
	\IPDS splits an input of length $n$ into $k$ independent subproblems of size $\Theta(n/k)$ whp.
	It then calls itself recursively until the base case size of $N_0$ is reached.
	Whp, this involves $\Oh{\log_k (n / N_0)}$ recursion layers, each taking time $\Oh{n}$ and requiring $\Oh{k}$ words of memory for a depth-first traversal.
	The base case \FY uses $\Oh{1}$ words of memory and takes time $\Oh{n'}$ for a subproblem of size $n'$ and in total $\Oh{n}$ for all subproblems.
\end{proof}

\section{Parallel algorithms}\label{sec:parallel}
In this section, we introduce \PIPDSFull (\algo{\textbf{P}IpScShuf}), a parallel variant of \IPDS.
It is obvious that after running \IPS (\ie a single recursion level of \IPDS) we can process the $k$ independent subproblems pleasingly in parallel --- this is one of the core insights of the original  \SandShufFull~\cite{DBLP:journals/ipl/Sanders98}.
Unfortunately, in our case, parallelizing the subproblems alone leads to a linear parallel depth, since the first \IPS execution requires $\Omega(n)$ sequential work.
Therefore, we also have to parallelize \IPS itself.
We focus on the parallelization of \RoughScatter which, in practice, accounts for the vast majority of work.

\subsection{Parallelizing RoughScatter}
%\SandShufFull already exploits that the assignment of elements into buckets is pleasingly parallel.
At heart, the parallel \ParRoughScatter runs the sequential~\RoughScatter concurrently on independent subproblems.
To this end, we exploit that \RoughScatter allows arbitrary gaps \emph{between} buckets.
Secondly, we can freely pause and resume after each assignment without additional overhead since the algorithm's state is fully captured by the buckets' pointer triples and the partition of the placed elements.

Analogously to \cref{subsec:rough-shuffle}, we first split the input into~$k$ buckets of roughly equal size.
In order to fork, we further split each bucket into two, and assign either half to one subtask.
Then each subtask either recursively continues splitting, or, if the subproblem is sufficiently small, runs the sequential \RoughScatter.
After both subtasks join, we merge the two halves of each bucket.
This involves only operations on the buckets' pointers and swapping the staged items of the first half to the second half.
Additionally observe that the first subtask ends if there exists a filled bucket $B^{(1)}_i$, and analogously $B^{(2)}_j$ for the second subtask.
Only with probability $1/k$, we have $i = j$, and thus, the merged bucket $B_j$ is full.
Otherwise, all merged buckets contain at least one staged item and we continue executing \RoughScatter.

\begin{observation}\label{obs:par-r}
	Since \ParRoughScatter applies \RoughScatter after each join, the number~$R$ of remaining staged items according to \cref{lem:r_bound} also holds for \ParRoughScatter.
\end{observation}

\begin{lemma}\label{lem:par-rough-scatter}
	For $n \gg k \log^3 k$ and $k^2 = \Oh{n}$, whp \ParRoughScatter has $\Oh{\sqrt{n k \log k}}$ parallel depth and needs $\Oh{n}$ work.
\end{lemma}

\begin{proof}
	Splitting $k$ buckets into $2k$ takes $\Oh{k}$ time.
	By \cref{obs:par-r}, \cref{lem:r_bound} bounds the number of staged items received from both subtasks to $\Oh{\sqrt{nk\log k}}$.
	This bounds from above the time required to swap elements during merging, as well as  the time to run \RoughScatter on the remaining staged elements after merging.
	To meet the prerequisites of \cref{lem:r_bound}, we choose a base case size of $N_0 = k^2$ and process smaller subproblem sequentially in time $\Oh{N_0}$.
	This leads to the following bound on the parallel depth~$D(n)$:
	\begin{align}
		D(n) &= \begin{cases}
					D(n/2) + \Oh{k + \sqrt{n k \log k}} & \text{if } n \ge N_0 \\
					\Oh{N_0} & \text{if } n < N_0
				\end{cases}\\
			&= \Oh{N_0 + \log({n}/{N_0}) k + \sqrt{n k \log k}}
			= \Oh{\sqrt{n k \log k}}
	\end{align}

	\noindent
	Analogously, we bound the work using the Master Theorem~\cite{DBLP:journals/sigact/BentleyHS80} for the following recursion:
	\begin{align}
	W(n) &= \begin{cases}
		2 W(n/2) + \Oh{k + \sqrt{n k \log k}} & \text{if } n \ge N_0 \\
		\Oh{N_0} & \text{if } n < N_0
	\end{cases} 
	= \Oh{n} \qedhere
\end{align}
\end{proof}

\subsection{Parallelizing FineScatter}
As we discuss in \cref{subsec:exp-parallel}, in practice, it is not necessary to parallelize \algo{FineScatter} due to its negligible impact on the total runtime.
Thus, in the following, we sketch just enough adoptions to reduce the parallel depth of \algo{FineScatter} to that of \ParRoughScatter.
By \cref{obs:par-r}, the analysis of \algo{FineScatter} in \cref{subsec:fine-scatter} remains valid after the execution of \ParRoughScatter.
By comparing with \cref{lem:par-rough-scatter}, we find that the parallel depth \ParRoughScatter dominates all sequential operations but \algo{TwoSweep}.

Recall that \algo{TwoSweep} shifts the boundaries of buckets to match the final bucket sizes in time $\Oh{k \sqrt{nk\log k}}$ whp.
Thus, we need to shave off only a factor of $\Theta(k)$ which is straightforward using standard parallelization techniques based on the following observation:
The prefix sum~$D_i$ defined in \cref{lem:mnf} can be interpreted as the number of elements that the end of bucket~$B_i$ needs to be shifted.
Thus, after computing~$(D_i)_i$ and placing one worker per bucket, each worker can shift elements in the appropriate direction.
To shift items between distant buckets, we run $\Oh{k}$ rounds.
The time per round is dominated by the largest swap of items over any bucket boundary which, in turn, is upper bounded by the maximal deviation $\Oh{\sqrt{n/k \log k}}$ whp in the first round.
Thus, \algo{TwoSweep} can be na\"ively parallelized with a parallel depth of $\Oh{\sqrt{nk\log k}}$ matching that of \ParRoughScatter.
This results in a trivial upper bound of work of $\Oh{k \sqrt{nk\log k}}$ matching the overestimation of \cref{cor:two-sweep-time} used for the sequential \algo{TwoSweep}.

\begin{theorem}
	\PIPDS has (whp) parallel depth $\Oh{\sqrt{n k / \log k} \log(n)}$, uses $\Oh{n \log_k (n)}$ work and $\Oh{k [\log_k(n) + \nproc]}$ words of memory where $\nproc$ is the number of parallel subtasks.
\end{theorem}

\begin{proof}
	The proof is analogous to the proof of~\cref{thm:ipds-runtime} by replacing the complexity measures of \IPS with \cref{lem:par-rough-scatter} followed by symbolic simplifications.
	The memory bound additionally accounts for $k$ bucket pointer triples per subtask.
\end{proof}

\section{Implementation}\label{sec:implementation}
Our implementations use \texttt{Rust}, a programming language with strong memory safety and parallelism guarantees.
While the code repository contains a number of prototypes, we consider the publicly exposed algorithms, such as \IPDS (\texttt{seq\_shuffle}) and \PIPDS (\texttt{par\_shuffle}), ready to be used in other projects.
To monitor the code quality, we rely on the strong static analysis tools and dynamic checks available in the Rust ecosystem.
We also use more than 80 tests that include statistical tests of the uniformity of the produced permutations (\eg the $1$ and $2$-independence of the output ranks).

\PIPDS uses  the work-stealing scheduler included in the \texttt{rayon}\footnote{\url{https://github.com/rayon-rs/rayon}} crate.
We exclusively use binary Fork-Join parallelism by means of the \texttt{rayon::join} function which requires no heap allocations after the worker pool was once initialized.
Given the widespread usage of \texttt{rayon}, it is very likely that the calling application already set up this pool.
Then, none of our algorithms cause any heap memory allocation.
A \texttt{rayon::join} incurs very little cost if both tasks are executed on the same worker.
Hence, we regularly define more than $2^{11}$ parallel subtasks --- allowing fine-grained work-balancing.
In case of a compatible pseudo-random number generator (requires \texttt{rand::SeedableRNG} trait), we use a deterministic sequence to derive the subtasks' generators from the provided generator.
Then, two runs from the same state yield the same permutation despite non-deterministic scheduling; this optional reproducibility can be crucial (\eg for debugging of the embedding code).

We empirically optimized the number of buckets~$k$ as $k = 64$ for medium sized inputs (below \qty{128}{\mebi\byte}) and $k = 256$ for larger instances.
This threshold is a compromise over various computer architectures with medium to low impact.
The vast majority of code is implemented in the \emph{safe} subset of Rust.
In \cref{sec:evaluation}, we use an optional highly optimized implementation of \RoughScatter that requires pointer arithmetic and memory accesses without explicit boundary checks which is considered \emph{unsafe} in Rust.
On the \texttt{x86} platform, the code executes $2 \lfloor 2 / \lceil \log_2(k) \rceil \rfloor$ (i.e. $20$ for $k=64$ and $16$ for $k=256$) random assignments without any branching instructions.
This allows a high utilization of the CPU's pipeline which is further increased by explicitly prefetching the memory locations.
Further, instead of using a standard \texttt{swap}$(x, y)$ with three move operations (namely $t \gets x,\ x \gets y,\ y\gets t$ where $t$ is a temporary storage), we use two temporary variables resulting in $2 + \epsilon$ data movements per assignment.
The resulting assignment process is at least five times faster than any weighted sampling strategy we experimented with (including fast implementations of \cite{DBLP:journals/mst/MatiasVN03} and various rejection schemes in spirit of \cite{DBLP:conf/esa/BerenbrinkHK0PT20}).

\IPDS and \PIPDS use a \FY implementation for instances below $2^{18}$ items that resorts to \qty{32}{\bit} arithmetic and often produces two indices from one random \qty{64}{\bit} word.

The repository includes highly optimized sequential and parallel reimplementations of \algo{MergeShuffle} which include similar techniques as above.
Their performance is incomparable with the original implementation\footnote{\url{https://github.com/axel-bacher/mergeshuffle}} which, on the one hand, includes handcrafted assembly code for merging, but, on the other hand, uses the \texttt{rdrand} instruction~\cite{intel-software-developer-manual} to acquire random bits;
depending on the specific processor, \texttt{rdrand} one to two orders of magnitude slower than \texttt{Pcg64Mcg}.~\cite{agner-fog-instr-tab,Route_2017}
Further, both choices are highly non-portable.
Overall our portable implementation using \texttt{Pcg64Mcg} is faster than the original.
Observe that \algo{MergeShuffle} has linear parallel depth since only independent subproblems are executed in parallel while the merging of the first recursion layer is purely sequential.

\section{Empirical evaluation}\label{sec:evaluation}
In this section, we investigate the performance of multiple shuffling algorithms for diverse settings.
If not stated differently, we use the following standard parameters:
\begin{itemize}
	\item 
	Measurements are collected on a machine with an \textsl{AMD EPYC 7702P processor} with 64 cores and 128 hardware threads, 512 GB RAM, running \textsl{Ubuntu 20.04.5}, \texttt{rustc 1.68.0-nightly (92c1937a9)} and \texttt{gcc-10}.
	We use standard release builds (\texttt{cargo --release}, \texttt{g++ -O3}) without machine-specific optimizations.
	
	\item 
	Experiments focus on the fast pseudo-random number generator, \texttt{Pcg64Mcg}, as this choice exposes overheads in the shuffling algorithms rather in the generators itself.
	In \cref{fig:rngs}, we report the performance for different generators.
	The relative performance of algorithms remains qualitatively similar for different randomness sources, though \IPDS and \PIPDS are less affected by the generator choice as \FY variants.
	
	\item 
	Experiments focus on \qty{64}{\bit} integers, which seems to be a typical index size in data sets of several \qty{100}{\gibi\byte}. 
	As indicated in \cref{fig:byte_scale}, the throughput (measured in bytes per second) increases for larger elements since the per-element overhead shrinks.
	Again, \IPDS and \PIPDS exhibit a smaller spread than \FY variants.
	
	\item All performance measurements reported are the mean of at least five runs.
	To reduce systematic errors, an individual run is the average of $N$ repetitions where $N$ is chosen such that the measured time exceeds \qty{100}{\milli\second}.
	Consecutive repetitions use different locations in a larger memory region to simulate a \emph{cold start} where the input is not already cached.
\end{itemize}

\subsection{Memory usage and allocation costs}\label{subsec:exp-memory}
\begin{figure}
	\begin{subfigure}{0.45\textwidth}
		\includegraphics[width=\textwidth]{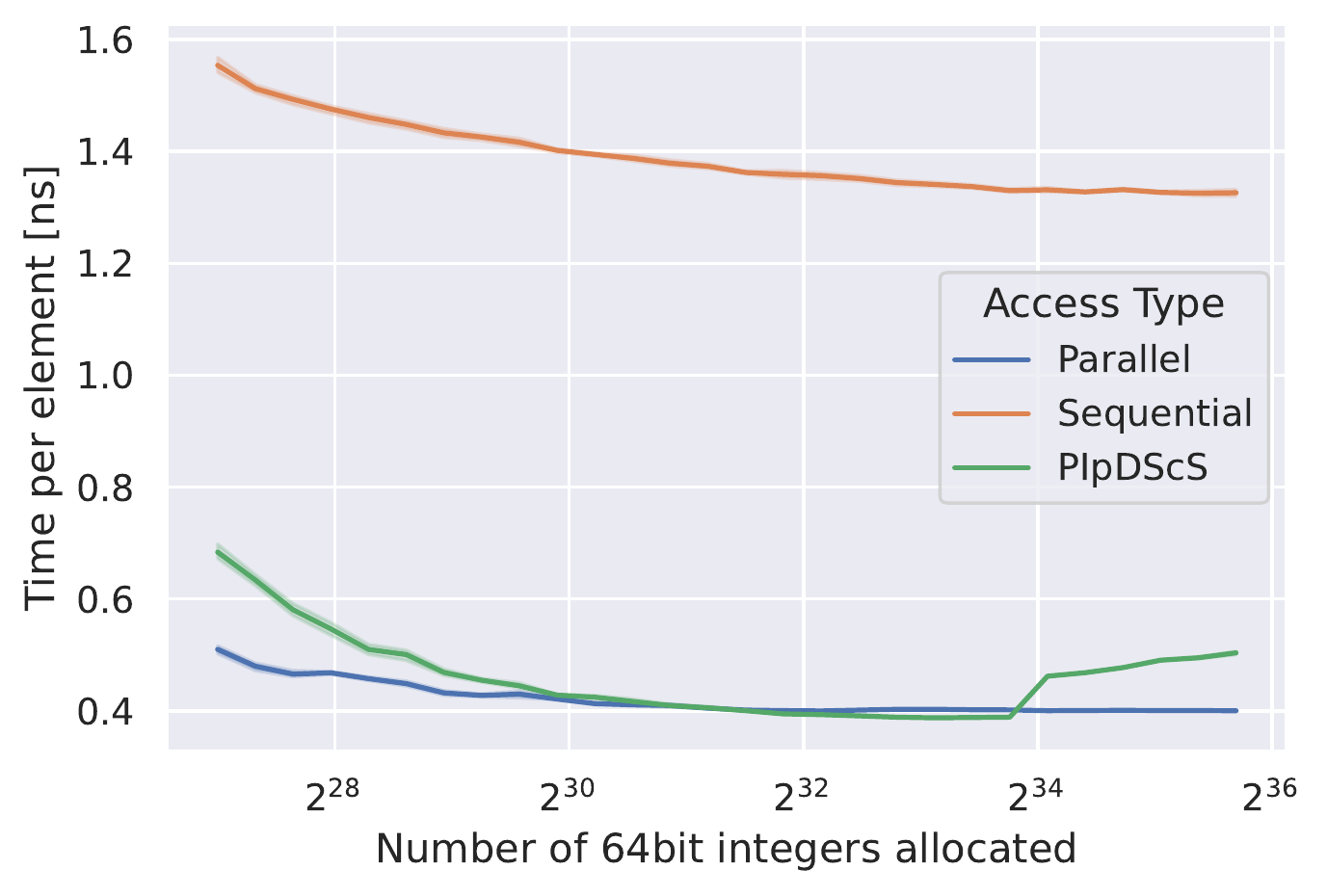}
		\caption{Runtime cost of memory allocation}
		\label{subfig:allocation-cost}
	\end{subfigure}\hfill
	\begin{subfigure}{0.45\textwidth}
		\includegraphics[width=\textwidth]{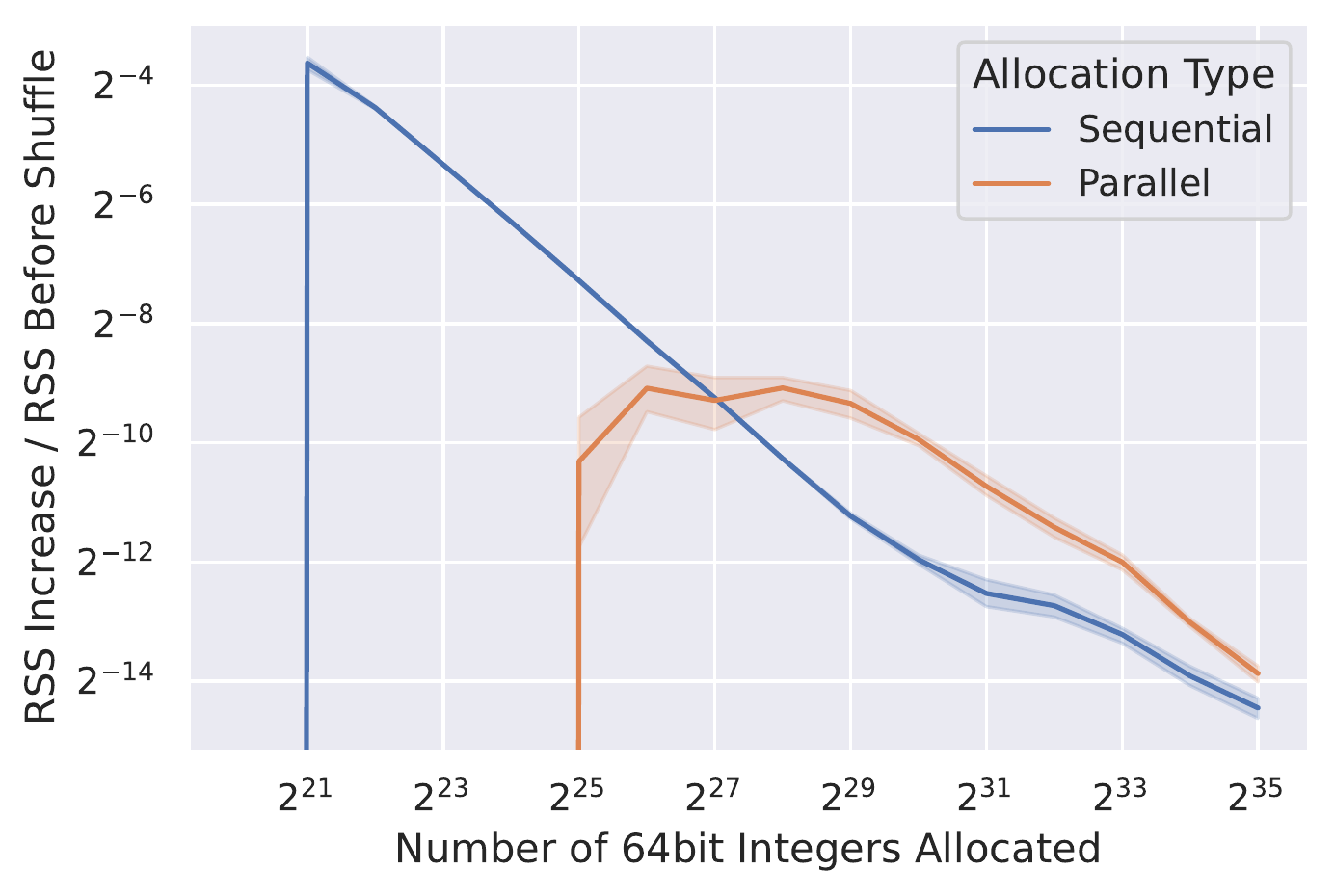}
		\caption{Memory usage of \PIPDS implementation}
		\label{subfig:memory-usage}
	\end{subfigure}
	\caption{Measurements of memory runtime costs and memory usage as described in \cref{subsec:exp-memory}.}
\end{figure}

One important motivation for this work is the runtime cost of allocating large amounts of memory which we quantify in \cref{subfig:allocation-cost}.
For each measurement, we obtain a certain amount of data using the low level \texttt{libc::malloc} instruction, initialize it, and then return the data using \texttt{libc::free}.
For large volumes, \texttt{malloc} requests the operating system to map a certain memory size into virtual memory.
Critically, the memory will only be backed by physical memory if it is actually accessed.
For this reason, we initialize the data twice, and subtract the second round from the first one.
The difference between both runs is the time it takes the system to provide the physical memory (without initializing it).
Our measurements also suggest that writing the values in parallel does not scale well --- despite an investment of 128 threads, we observe only a speedup of $4.9$ for the initialization which is reduced to $3.1$ for the whole process since \texttt{malloc} and \texttt{free} are sequential.
For reference, we included the runtime of \PIPDS and find that shuffling the data in parallel takes roughly as long as to acquire the memory needed to store a copy --- without copying it.

In \cref{subfig:memory-usage}, we report the effective memory usage of \PIPDS.
We start a dedicated process for each run and measure the \emph{maximal resident set size} (RSS) of the process (\ie the maximal amount of memory that was physically backed at any time during the execution).
We measure the RSS before and after the invocation of \PIPDS and report the relative growth.
As already discussed in \cref{sec:implementation}, \texttt{rayon}'s worker pool needs to be initialized once.
If we allocated the data sequentially, \PIPDS implicitly sets up a pool resulting in 1631 heap operations to reserve a total of \qty{923}{\kibi\byte}.
After a parallel allocation, on the other hand, no heap operations are carried out.
Even then, we observe a small increase of the RSS for large inputs --- this seems to be caused by growing stack memory of the 128 active threads.
In this case, the largest observed growth is \qty{0.2}{\percent} and diminishes for very large inputs.

\subsection{Performance overview}\label{subsec:exp-individual}
\begin{figure}
	\includegraphics[width=\textwidth]{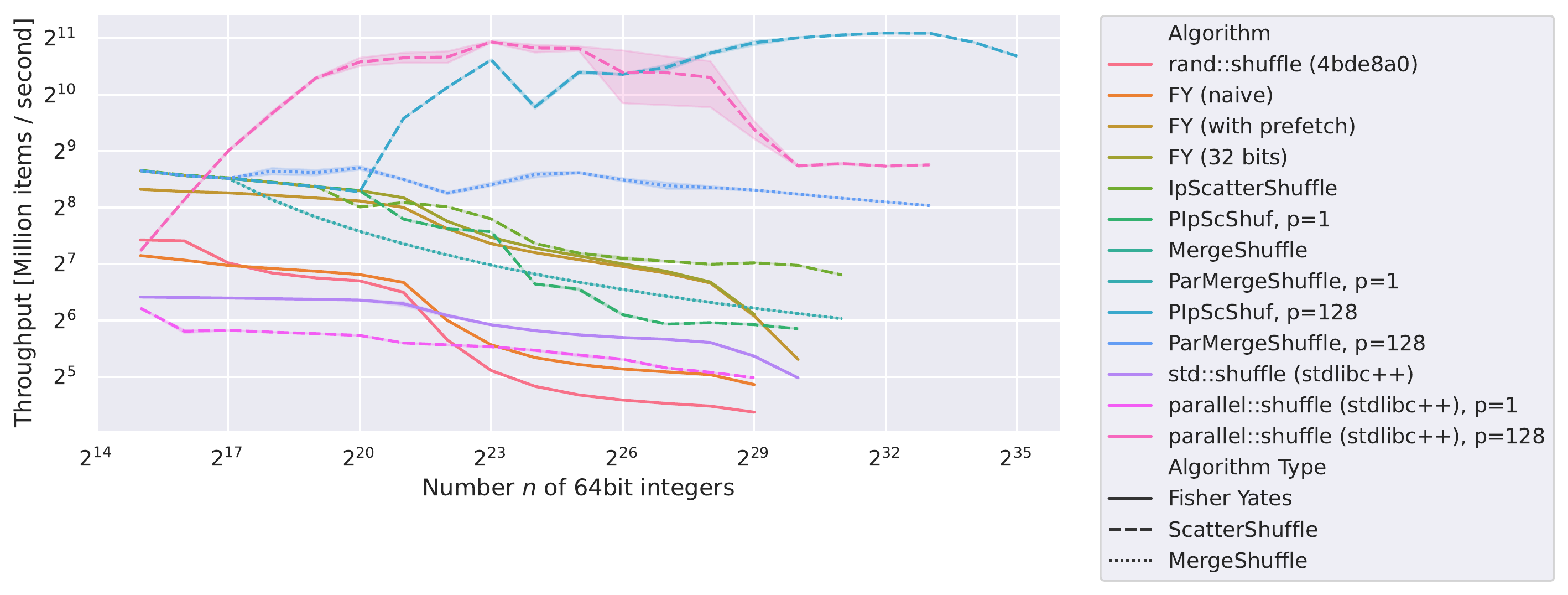}
	\caption{
		Performance of several shuffling algorithms with a time budget of \qty{30}{\second} per run.
		\FY, \IPDS, and \PIPDS are our own implementations.
		\texttt{std::shuffle} and \texttt{parallel:shuffle} are implemented in C++.
		For parallel algorithms, we indicate the number of cores available as $p$.
	}
	\label{fig:size_scale}
\end{figure}

\begin{figure}
	\includegraphics[width=\textwidth]{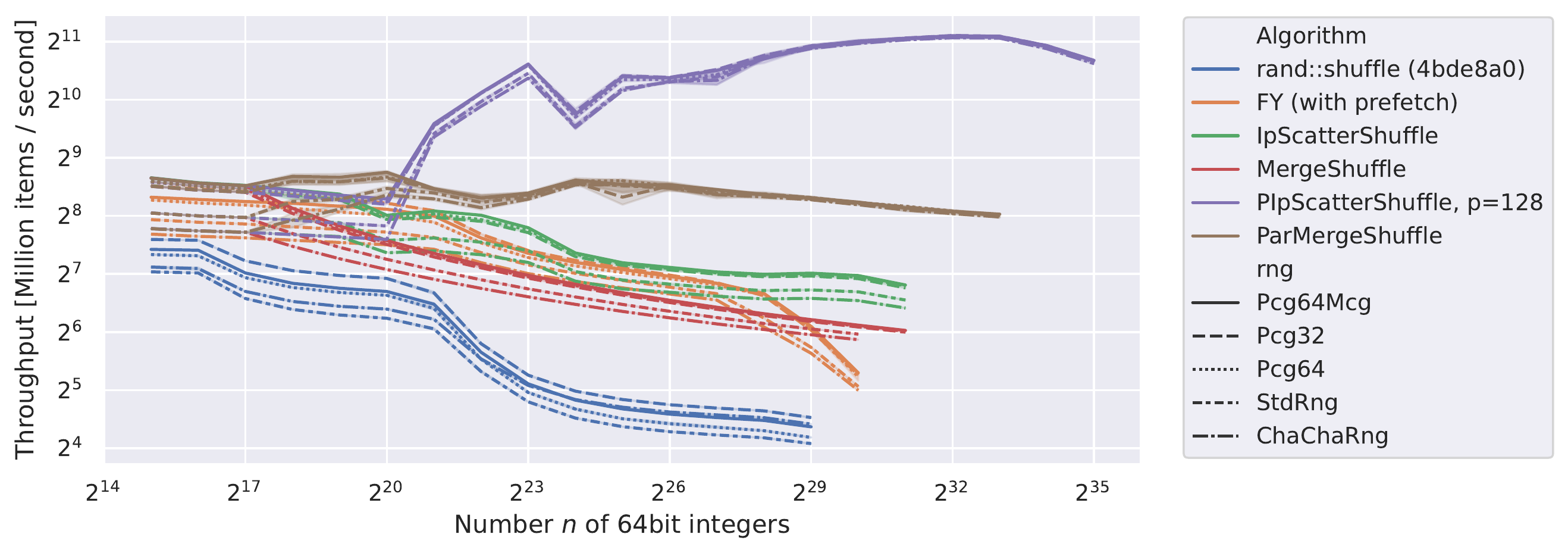}
	\caption{
		Performance of selected algorithms with different pseudo-random number generators.
	}
	\label{fig:rngs}
\end{figure}

\begin{figure}
	\includegraphics[width=\textwidth]{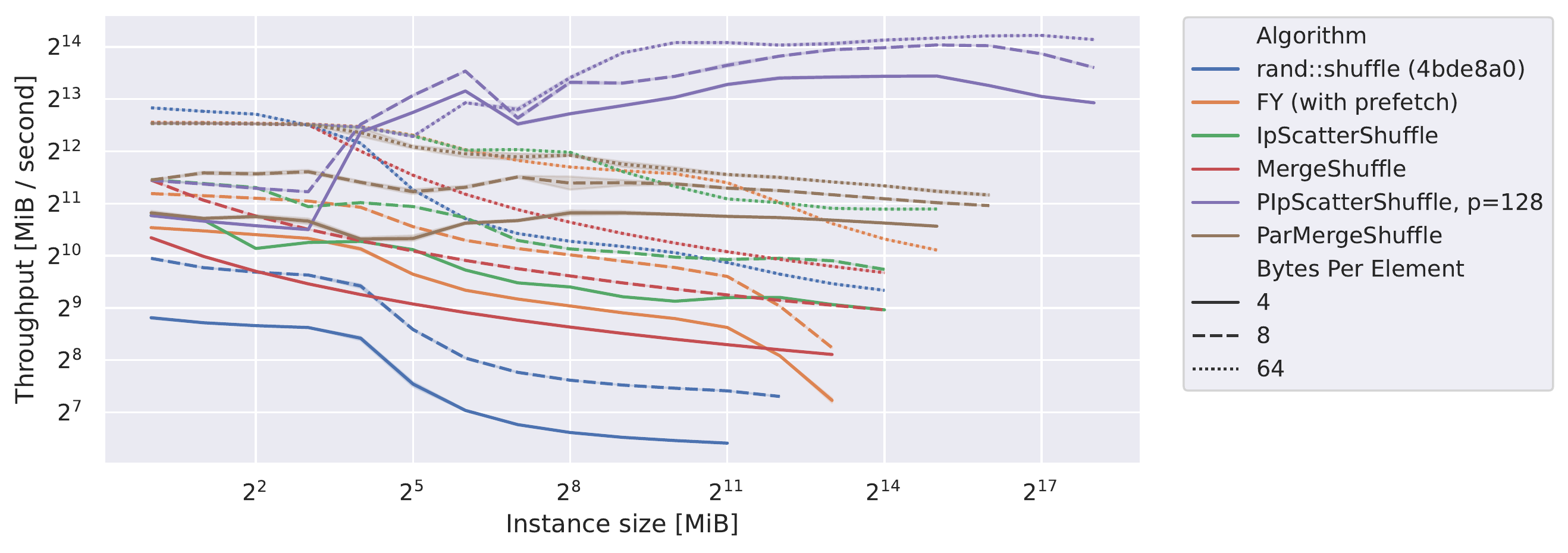}
	\caption{
		Performance of selected algorithms with different element sizes.
	}
	\label{fig:byte_scale}
\end{figure}

In \cref{fig:size_scale}, we report the performance of several shuffling implementations.
For each run, we set a timeout of \qty{30}{\second} and stop a graph after its algorithm hit said budget.
The only exception is our \PIPDS implementation with a runtime of \qty{20.8}{\second} for the largest data point.

The two fastest algorithms are \texttt{parallel::shuffle} (a C++ implementation of \SandShufFull included in \texttt{stdlibc++}) and our \PIPDS.
For relatively small inputs \SandShufFull is faster than \PIPDS which takes the lead for inputs larger than \qty{256}{\mebi\byte}.

All algorithms exhibit deteriorating throughput for larger inputs.
This is remarkable for \FY derivatives which have a predicted linear runtime.
Their slowdown can be attributed to cache misses and related effects of the memory hierarchy.
By comparing \texttt{FY (naive)} with \texttt{FY (with prefetch)}, we demonstrate that memory latency can be hidden  by explicitly prefetching memory locations ahead of time.\footnote{We use a ring buffer to generate and prefetch random indices 16 swaps prior to the actual swap.}
However, prefetching only helps to simulate a slightly larger cache and we still observe a significant performance drop around \qty{2}{\gibi\byte}.
The pronounced slowdown of the de facto standard Rust shuffle (\texttt{rand::shuffle}) is caused by a recent optimization for small ranges. It maximizes the number of random indices obtained from a single random \qty{32}{\bit} word and becomes less efficient for larger input sizes.

As predicted, all implementations based on \SandShufFull exhibit a $\log_k(n)$ dependency in their runtime; this is especially visible for \PIPDS when executed on a single core.
The parallel implementation is more affected than the dedicated sequential variant since each recursion adds additional overheads for splitting larger subproblems into independent tasks.

Due to incompatible software libraries, we were unable to include the performance measurements of \cite{DBLP:conf/apocs/GuOS21} in our campaign.
However, \cref{fig:hosts} reports a higher throughput for \PIPDS on a quad core laptop (i7-8550U) than \cite[Table 6]{DBLP:conf/apocs/GuOS21} for a quad-socket server (E7-8867 v4) with 72 ~cores --- while \PIPDS use less memory and includes the computation time of random bits.

\subsection{Parallel execution of sequential algorithms}
\begin{figure}
	\includegraphics[width=\textwidth]{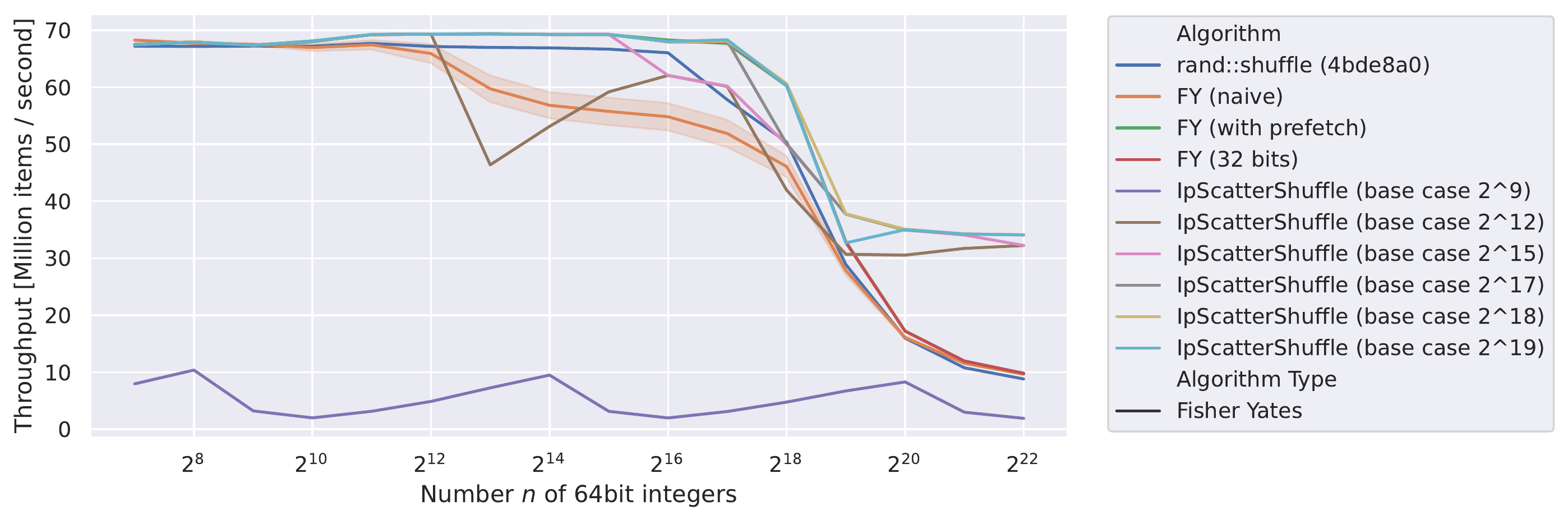}
	\caption{
		Performance of several sequential shuffling algorithms run in parallel on different data.
	}
	\label{fig:par_seq_scale}
\end{figure}

When selecting an appropriate base case algorithm for \PIPDS, \cref{fig:size_scale} can be misleading as it reports the performance of sequential algorithms executed in isolation.
In this setting, the studied algorithm has much more resources at its disposal compared to the case where several instances are executed in parallel on independent data.
This might also be relevant in different scenarios, \eg if (unbeknownst to the user) a sequential numerical simulation, which is executed in the ``cloud'', gets colocated with other memory-intensive tasks.

\Cref{fig:par_seq_scale} is recorded similarly to \cref{fig:size_scale} with the difference that we execute 128 independent tasks in parallel and report the mean of their individual runtime as a single run.
To avoid scheduling artifacts, we discard and repeat any run in which the wall-time of the experiment (\ie from the start of the first thread to the termination of the last) is \qty{20}{\percent} larger than the mean runtime of the individual tasks.

In this setting, we observe that memory becomes the dominating bottleneck; with minor exceptions (such as too small base case sizes), all algorithms exhibit roughly the same performance for instances below \qty{1}{\mebi\byte} (the CPU has \qty{256}{\mebi\byte} L3 cache that is now shared among 128 threads).
For larger instances, \IPDS variants are more than 3 times faster than the \FY variants.
Also observe that the contributions of implementation details, such as prefetching or base case size, pale in comparison the importance of memory locality. 

\subsection{Parallel scaling}\label{subsec:exp-parallel}
\begin{figure}
	\begin{subfigure}[t]{0.49\textwidth}
		\includegraphics[width=\textwidth]{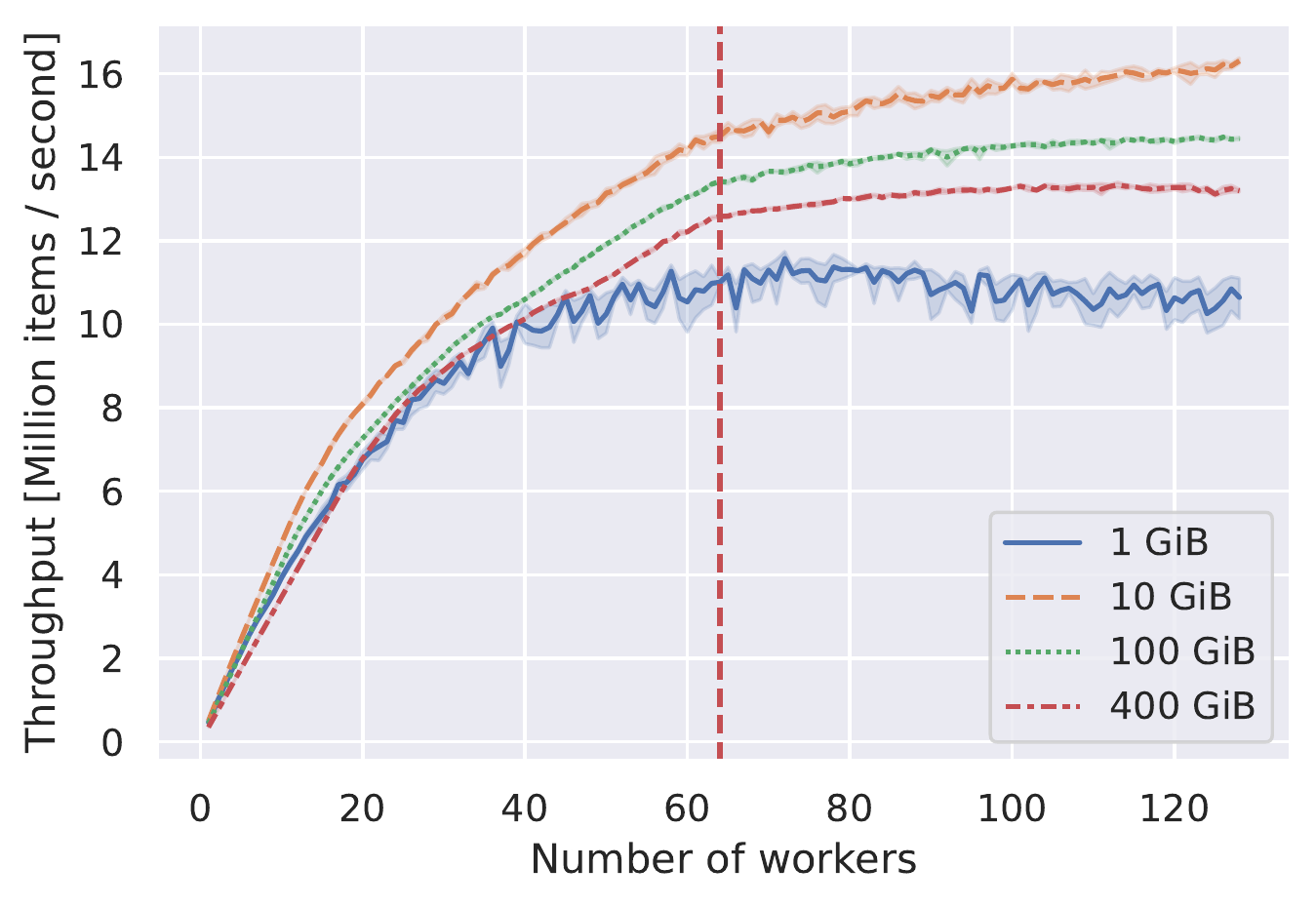}
		\caption{
			Strong scaling of \PIPDS.
			The vertical line indicates the number of physical cores.
		}
		\label{subfig:par-strong-scale}
	\end{subfigure}\hfill
	\begin{subfigure}[t]{0.49\textwidth}
		\includegraphics[width=\textwidth]{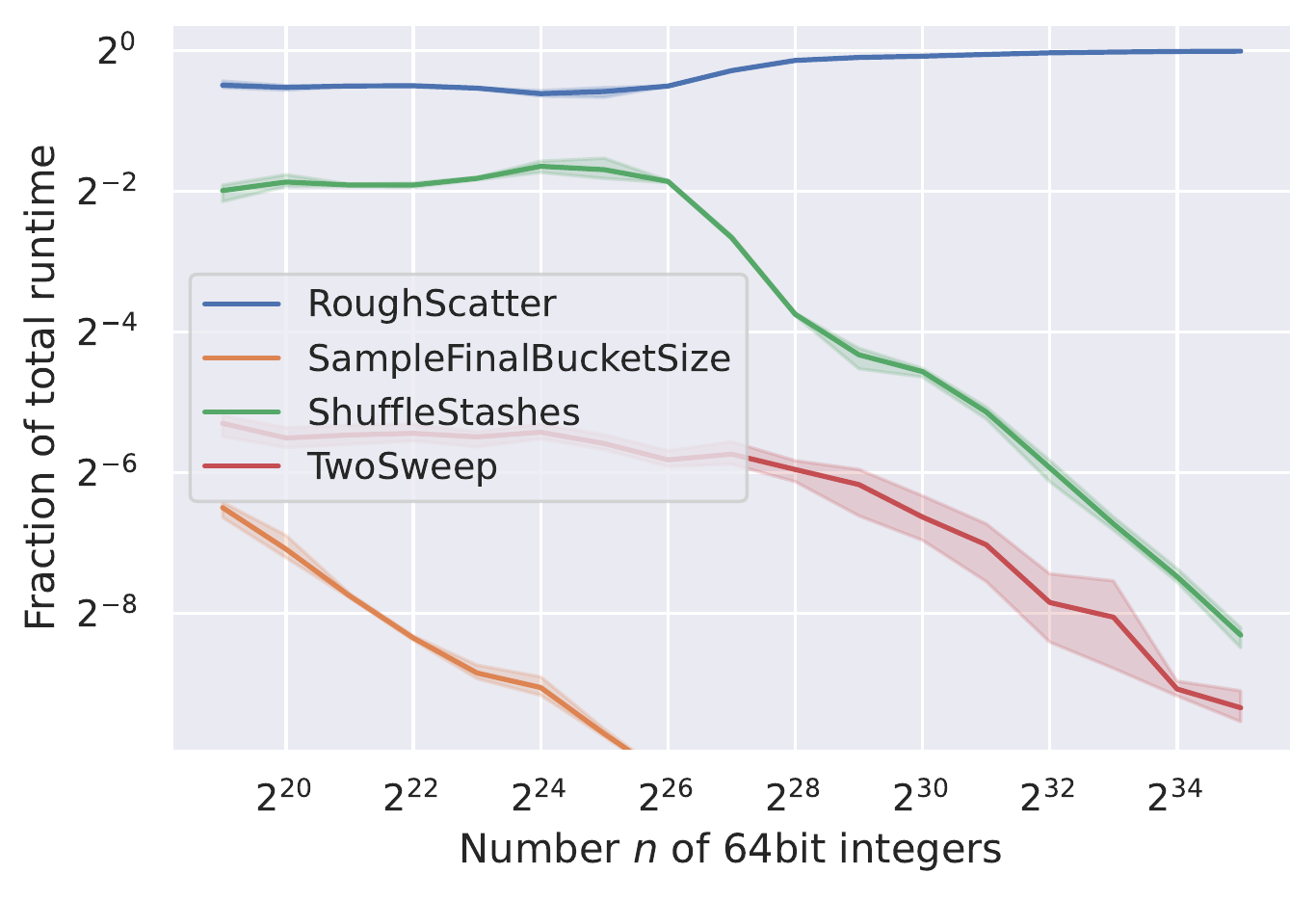}
		\caption{Fraction of runtime of the first recursion layer of \PIPDS with $128$ cores.}
		\label{subfig:parallel_regions}
	\end{subfigure}
	\caption{Parallel performance of \PIPDSFull.}
\end{figure}
To quantify the parallel speedup of \PIPDS, we carry out a strong scaling experiment as follows.
For fixed input sizes, we execute \IPDS and the fastest \FY implementation as base lines and then profile \PIPDS for an increasing number of workers.
In \cref{subfig:par-strong-scale}, we report the parallel speedup over the fast sequential implementation (\IPDS in all cases).

For data set sizes of \qty{10}{\gibi\byte} and larger, the speedup over the fastest sequential solution approaches up to 16.
The self speedup is larger (\eg 32.2 for the \qty{100}{\gibi\byte}), indicating a good scalability that is somewhat offset by the additional overhead of the parallel implementation.
For the same instance, \PIPDS is 142 times faster than \FYFull.

We see a substantial increase in speed until the number of workers matches the number of physical cores; using virtual cores (simultaneous multi-threading) has little impact.
This is to be expected since shuffling is memory-bound and virtual cores primarily help to saturate arithmetic units of super-scalar processors, but do not affect memory performance.

\subsection{Performance on different machines}
\begin{figure}
	\includegraphics[width=\textwidth]{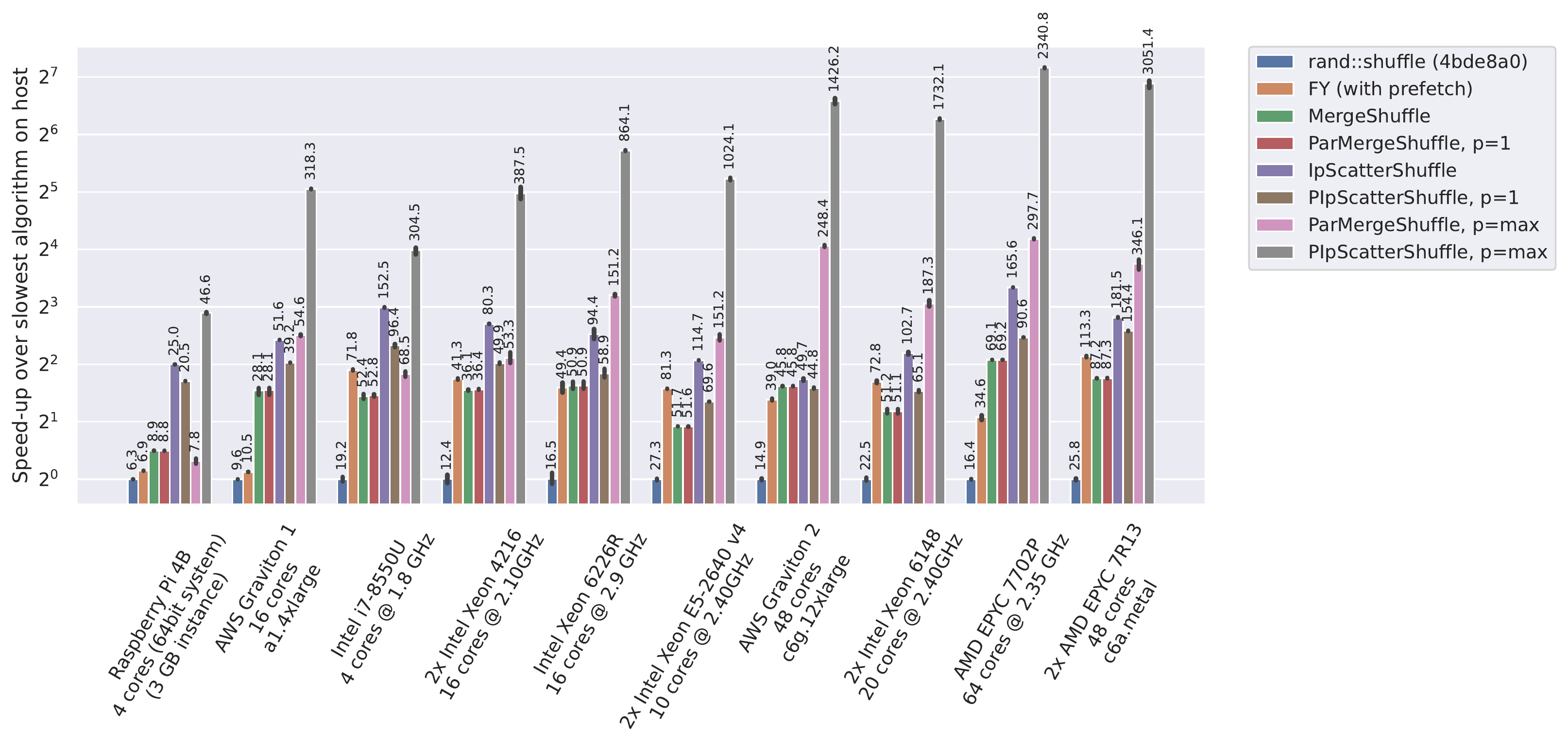}
	\caption{
		Performance of selected algorithms on different computers shuffling \qty{64}{\bit} integers.
		The length of a bar corresponds to the speedup compared to the slowest algorithm on that system.
		The numbers above a bar indicate the absolute through-put in million elements per second.
	}
	\label{fig:hosts}
\end{figure}

To verify that our empirical findings are representative for modern computers, we quantify the performance of shuffling on different machines in \cref{fig:hosts}.
The machines are diverse and ranging from a single-board computer, over a laptop, to dual-socket servers (covering more than two orders in magnitude in purchase prize).
They use different instruction sets, micro-architectures, processor manufactures, and core counts.

To accommodate most systems, we selected an instance size of \qty{10}{\gibi\byte} with the exception of the Raspberry PI~4B which features only \qty{4}{\gibi\byte} of main memory. 
The measurements consist of runs of sequential algorithms, sequential runs of parallel algorithms (indicated by $p=1$), and parallel runs with one worker per hardware thread (indicated by $p=\text{max}$).
The maximal throughput of the fastest system is 50 times higher than that of the slowest system.
In all cases \texttt{rand::shuffle} is the slowest contender, \IPDS the fastest sequential implementation, and \PIPDS the overall fastest solution.

\subsection{Relative performance of subproblems}
When designing, implementing, and analyzing \IPDS and \PIPDS we focused on the fast opportunistic \RoughScatter which then requires the additional \algo{FineScatter} post-processing to deal with the few remaining items.
While we heavily optimized \RoughScatter, we opted for simple and easy to implement solutions in \algo{FineScatter}.

To empirically support this design decision, we measure the runtime of the first recursion layer of \PIPDS for a wide range of input sizes.
In \cref{subfig:parallel_regions}, we then report the relative wall time of the four sub-algorithms that constitute said layer.
In our implementation only \RoughScatter is executed in parallel with 128 workers available while the remaining parts are sequential algorithms.
Despite the asymmetry in available workers, we observe that for $n \ge 2^{27}$ \ParRoughScatter accounts for more than \qty{90}{\percent} (\qty{99}{\percent} for $n \ge 2^{33}$) of the runtime.
This supports our design decisions since optimizing \algo{FineScatter} leads to diminishing results.

\section{Quantifying the hidden constants}
\begin{figure}
	\begin{center}
		\includegraphics[width=0.45\textwidth]{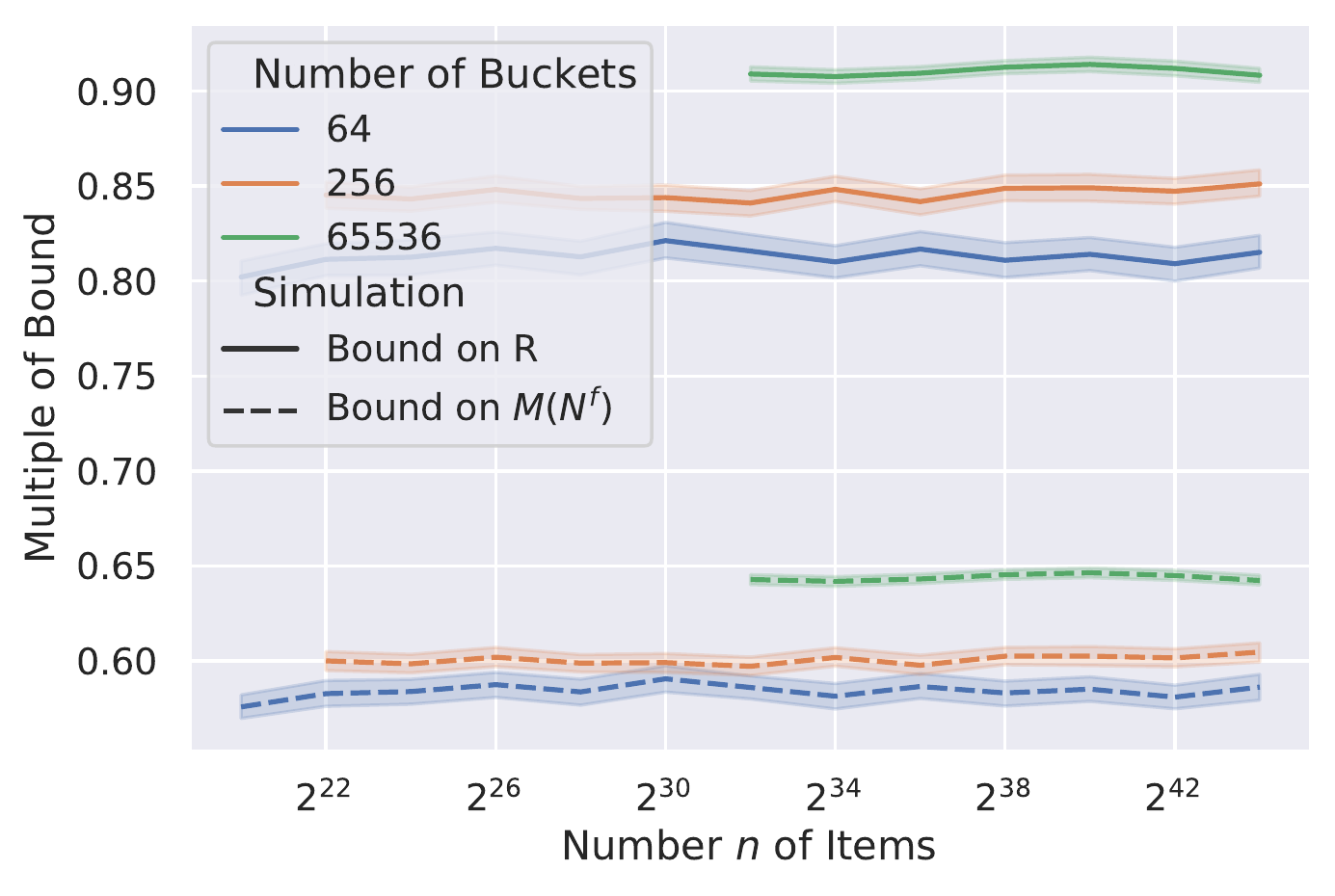}
	\end{center}

	\caption{Simulation of \cref{lem:r_bound} and \cref{cor:two-sweep-time}.}
	\label{subfig:hidden-costs}
\end{figure}
In \cref{lem:r_bound}, we bound the number $R \le \sqrt {2 n k \log k}$  of items that remain staged after \RoughScatter whp.
Additionally, in \cref{cor:two-sweep-time}, we bounded the complexity of \algo{TwoSweep} by $\Oh{k \sqrt{nk\log k}}$.
To provide empirical evidence and study the hidden constants, we simulate both processes.
In \cref{subfig:hidden-costs}, we report the mean over $1000$ independent runs divided by $\sqrt {2 n k \log k}$ and $k \sqrt{nk\log k}$ respectively.
Recall that our implementations use $k = 64$ and $k=256$; we additionally simulated $k=2^{16}$ as an accommodating upper bound for the foreseeable future.
The small growth in~$k$ visible in \cref{subfig:hidden-costs} is due to the small $k$ values.
Simulations with $k$ up to $2^{28}$ agree with \cref{lem:r_bound} and approach a factor of $0.66$ in \cref{cor:two-sweep-time}.

\bibliography{references,dblp-downloaded}

\end{document}